\documentclass[conference]{IEEEtran}

\usepackage{amsmath, amsthm, amssymb}
\usepackage{latexsym}
\usepackage{graphicx}
\usepackage{verbatim}
\usepackage{mathrsfs}
\usepackage{float}
\usepackage[lofdepth, lotdepth]{subfig}
\usepackage{array}
\newcolumntype{P}[1]{>{\centering\arraybackslash}p{#1}}
\usepackage{url}
\usepackage{algorithm}
\usepackage[noend]{algorithmic}
\usepackage{cite}
\usepackage[table]{xcolor}
\usepackage{enumerate}
\usepackage{eucal}
\usepackage{stmaryrd}
\usepackage{mathtools}
\usepackage{pgf}
\usepackage{tikz}
\usetikzlibrary{arrows,automata}
\usepackage[latin1]{inputenc}
\usetikzlibrary{automata,positioning}
\usepackage{mathtools}
\DeclarePairedDelimiter\ceil{\lceil}{\rceil}
\DeclarePairedDelimiter\floor{\lfloor}{\rfloor}

\theoremstyle{plain}
\newtheorem{theorem}{Theorem}

\newtheorem{lemma}{Lemma}

\newtheorem{corollary}{Corollary}

\theoremstyle{definition}
\newtheorem{definition}{Definition}
\newtheorem{example}{Example}
\newtheorem{remark}{Remark}

\newcommand{\B}{{\mathcal B}}
\newcommand{\C}{{\mathcal C}}
\newcommand{\D}{{\mathcal D}}

\DeclareMathAlphabet{\mathbfsl}{OT1}{ppl}{b}{it} %{OT1}{cmr}{bx}{it}

\newcommand{\bS}{{\mathbfsl{S}}}

\newcommand{\bL}{\mathbfsl{L}}

\newcommand{\bu}{{\mathbfsl u}}
\newcommand{\bv}{{\mathbfsl v}}

\newcommand{\by}{{\mathbfsl y}}

\newcommand{\bc}{{\mathbfsl c}}

\newcommand{\bw}{{\mathbfsl{w}}}

\newcommand{\bx}{{\mathbfsl{x}}}
\newcommand{\bz}{{\mathbfsl{z}}}

\newcommand{\bT}{{\mathbfsl{T}}}

\newcommand{\bbZ}{{\mathbb Z}}

\newcommand{\ppmod}[1]{~({\rm mod~}#1)}

\renewcommand{\ge}{\geqslant}
\renewcommand{\le}{\leqslant}

\newcommand{\et}{{\emph{et al.}}}

\newcommand{\enc}{\textsc{Enc}}
\newcommand{\dec}{\textsc{Dec}}

\IEEEoverridecommandlockouts
\begin{document}

\pagestyle{empty}

%\title{Codes for Efficient DNA Synthesis with Biochemical Constraints and Error Correction Codes\\[-3mm]}
\title{Constrained Coding for Composite DNA:\\ Channel Capacity and Efficient Constructions\\[-3mm]}

\author{\IEEEauthorblockN{Tuan Thanh Nguyen\IEEEauthorrefmark{1},
Chen Wang\IEEEauthorrefmark{1}\IEEEauthorrefmark{2},
Kui Cai\IEEEauthorrefmark{1},
Yiwei Zhang\IEEEauthorrefmark{2}, 
and Zohar Yakhini \IEEEauthorrefmark{3}\IEEEauthorrefmark{4}}\\[-3mm]

%and Kees A. Schouhamer Immink\IEEEauthorrefmark{3}}\\[-3mm]
\IEEEauthorblockA{
\IEEEauthorrefmark{1}
Science, Mathematics and Technology Cluster, Singapore University of Technology and Design, Singapore 487372\\
%Singapore University of Technology and Design, Singapore 487372\\
\IEEEauthorrefmark{2}
School of Cyber Science and Technology, Shandong University, Qingdao, Shandong, 266237, China\\
%Turing Machines Inc, Willemskade 15d, 3016 DK Rotterdam, The Netherlands\\
\IEEEauthorrefmark{3}
School of Computer Science, Herzliya Interdisciplinary Center, Herzliya, Israel\\
\IEEEauthorrefmark{4}
Faculty of Computer Science, Technion - Israel Institute
of Technology, Haifa 3200003, Israel\\
Emails: \{tuanthanh\_nguyen, cai\_kui\}@sutd.edu.sg, cwang2021@mail.sdu.edu.cn, \\ywzhang@sdu.edu.cn, zohar.yakhini@gmail.com \\[-4mm] %immink@turing-machines.com, j.h.weber@tudelft.nl
%immink@turing-machines.com
}
}

\maketitle

%SECTION ABSTRACT
\hspace{-3mm}\begin{abstract}

\emph{Composite DNA} is a recent novel method to increase the information capacity of DNA-based data storage above the theoretical limit of $2$ bits/symbol. In this method, every composite symbol does not store a single DNA nucleotide but a mixture of the four nucleotides in a predetermined ratio. By using different mixtures and ratios, the alphabet can be extended to have much more than four symbols in the naive approach. While this method enables higher data content per synthesis cycle, potentially reducing the \emph{DNA synthesis} cost, it also imposes significant challenges for accurate \emph{DNA sequencing} since the base-level errors can easily change the mixture of bases and their ratio, resulting in changes to the composite symbols. %To overcome this challenging issue, one may use constrained codes to reduce the occurrence of errors in both synthesis and sequencing process. However, due to the special structure of the alphabet symbols in composite DNA, 
\vspace{1mm}

With this motivation, we propose efficient constrained coding techniques to enforce the biological constraints, including the \emph{runlength-limited constraint} and the \emph{${\tt G}{\tt C}$-content} constraint, into every DNA synthesized oligo, regardless of the mixture of bases in each composite letter and their corresponding ratio. This approach plays a crucial role in reducing the occurrence of errors. Our goals include computing the capacity of the constrained channel, constructing efficient encoders/decoders, and providing the best options for the composite letters to obtain capacity-approaching codes. For certain codes' parameters, our methods incur only one redundant symbol. %To the best of our knowledge, no such codebooks are known prior to this work.

% in the synthesis and sequencing processes

%one may use constrained codes to enforce the biological constraints into every DNA oligo. By doing so (as previously discussed), the occurrence of errors in the synthesis process can be significantly reduced. 

%Advances in synthesis and sequencing technologies have made DNA macromolecules an attractive medium for digital information storage. The main challenge of making DNA storage systems competitive relative to existing storage technologies is the synthesis cost.

%DNA macromolecules have been found as a promising medium for digital information storage, which has extremely high data storage density

%However, to fully realize this potential, the synthesis and sequencing technologies must become significantly more cost/time-effective. Particularly, the crucial challenges include the high cost of DNA synthesis

%Composite DNA is a recent method to increase the base alphabet size in DNA-based data storage

%To allow higher potential information capacity, ... introduced the composite DNA synthesis method recently. In this method, the multiple copies created by the standard DNA synthesis method are utilized to create composite DNA symbols, defined by a mixture of DNA bases and their ratios in a specific position of the strands. By defining different mixtures and ratios, the alphabet can be extended to have more than four symbols. 

\end{abstract}

% section INTRODUCTION
\section{Introduction}
Advances in synthesis and sequencing technologies have made DNA macromolecules an attractive medium for digital information storage \cite{church:2012, goldman:2013, Yazdi:2015, erlich:2017}. The crucial challenges, which make DNA storage still far from being practical in storing large data, are due to its high cost and low speed for DNA synthesis and sequencing \cite{hec:2019, Milen:2024}. Particularly, DNA synthesis remains the most costly and time-consuming part of DNA storage, motivating recent research on the efficient synthesis of DNA \cite{Lenz:2020,HM:2023, Sini:2023, TTDNA, immink:2024, tu:2024, thanh:dna, dube:2019, caikui:ep}. 
%Different approaches have been proposed, such as efficient coding methods that focused on reducing the overall number of synthesis cycles, subsequently reducing the overall cost and time of DNA synthesis \cite{Lenz:2020,HM:2023, Sini:2023, TTDNA, immink:2024}, or those focused on increasing the amount of data that can be encoded into each fixed-length DNA oligo \cite{tu:2024, thanh:dna, dube:2019, caikui:ep}. %(for example, refer to Table I in \cite{caikui:ep} for a summary of different information rates obtained in recent DNA storage experiments). 
In most common setting, the digital information is converted to quaternary sequences using the standard DNA alphabet $\{{\tt A}, {\tt T}, {\tt C},{\tt G}\}$, which has a theoretical limit of
$\log_2 4 = 2$ bits/symbol. Since redundancy is required to enhance constrained properties and error-correction capability, the overall information rate is significantly reduced from this limit. Surprisingly, there is a recent novel method, proposed by L. Anavy \et{} \cite{zohar:nature}, that use {\em composite DNA} to increase the information capacity of DNA-based data storage above the theoretical limit of $2$ bits/symbol (see Table I, \cite{zohar:nature}). 

%For example, the information
%rate in [4] and [6] is at most $\log_2 3 \approx 1.58$

%This paper studies constrained coding problems that are motivated by the recent novel method, namely {\em composite DNA}, as proposed by L. Anavy \et{} \cite{zohar:nature} to increase the information capacity of DNA-based data storage above the theoretical limit of $2$ bits/symbol in the naive approach using standard DNA symbols ${\tt A}, {\tt T}, {\tt C}, {\tt G}$. We next briefly describe 
In composite DNA, each composite letter in a specific position consists of a mixture of the four DNA nucleotides in a predetermined ratio, i.e. it can be abstracted as a quartet of probabilities ${p_{\tt A},p_{\tt T},p_{\tt C},p_{\tt G}}$, in which $p_{\tt A}  + p_{\tt T}  + p_{\tt C}  + p_{\tt G}  = 1$. By using different mixtures and ratios, the alphabet can be extended to have much more than four symbols. On the other hand, to identify a composite letter, it is required to sequence a sufficient number of reads and then to estimate $p_{\tt A},p_{\tt T},p_{\tt C},p_{\tt G}$ in each position. While this method enables higher data content per synthesis cycle (equivalently, allows higher information capacity), it also imposes significant challenges for accurate DNA sequencing and requires a much larger {\em coverage} (or {\em sequencing depth}) \cite{zohar:nature, anavy:2024, Cohen:2024, shortmer:2024}. For example, in \cite{zohar:nature}, L. Anavy \et{} encoded 6.4 MB into composite DNA using the alphabet $\Sigma=\{{\tt A}, {\tt T}, {\tt C}, {\tt G}, {\tt M}, {\tt K}\}$, where two composite letters ${\tt M}, {\tt K}$ are defined as ${\tt M}=(0.5,0,0.5,0)$ and ${\tt K}=(0,0.5,0,0.5)$, i.e. ${\rm M}$ consists of a mixture of ${\tt A}$ and ${\tt C}$ in a ratio of $50\%:50\%$ while ${\tt K}$ consists of a mixture of ${\tt G}$ and ${\tt T}$ in a ratio of $50\%:50\%$. The information capacity is increased to $\log_2 6$ bits/symbol, and the authors managed to use $20\%$ fewer synthesis cycles per unit of data, as compared to previous reports. However, fully successful decoding in \cite{zohar:nature} required coverage of 100 reads (roughly increases by 10 times as required in standard DNA synthesis \cite{caikui:ep}). To overcome this challenge, several constructions of error correction codes (ECCs) have been proposed, specifically targeting the composite DNA method \cite{shortmer:2024, Walter:2024}. On the other hand, there is no efficient construction of constrained codes for composite DNA. 
\vspace{0.11mm} 

%In [9], coverage of more than 128 reads is required to guarantee at least 99\% decoding successfully.
%In this method, every composite symbol does not store a single DNA nucleotide but a mixture of the four nucleotides in a predetermined ratio. By using different mixtures and ratios, the alphabet can be extended to have much more than four symbols in the naive approach.
With this motivation, we propose efficient constrained coding techniques to enforce the biological constraints, including the \emph{runlength-limited constraint} and the \emph{${\tt G}{\tt C}$-content} constraint, into every DNA synthesized oligo, regardless of the mixture of bases in each composite letter and their corresponding ratio, and regardless of the positions of the composite letters in each data sequence. This approach plays a crucial role in reducing the occurrence of errors \cite{Yazdi:2015, hec:2019, Milen:2024, thanh:dna, xhe:2022}. Due to the unique structure of the alphabet symbols in composite DNA, such a constrained coding problem is much more challenging as compared to the standard DNA case, even in the most fundamental setting with one additional composite letter. Our goals include computing the capacity of the constrained channel, constructing linear-time encoders/decoders, and providing the best options for the composite letters to obtain capacity-approaching codes, particularly focusing on the parameters where the channel capacity is strictly larger than 2 bits/symbol.   

%We first go through certain notation and define the problem. 
%(see Example~\ref{exp1} and Figure~\ref{fig1}).
\section{Preliminary}
%\section{Definitions and Problem Statement}

%PRELIMINARY
Given two sequences $\bx$ and $\by$, we let $\bx\by$ denote the {\em concatenation} of the two sequences. In the special case where $\bx, \by$ are both of length $n$, we use $\bx||\by$ to denote their {\em interleaved sequence} $x_1y_1x_2y_2\ldots x_ny_n$. For two positive integers $m < n$, we let $[m,n]$ denote the set $\{m,m + 1,\ldots,n\}$ and we define $\bx_{[m,n]}=x_m x_{m+1} \ldots x_n$.
%\vspace{0.1mm}

%We first go through certain notations and define the problem. 
%We use $\Sigma_k=\{{\tt A}, {\tt T}, {\tt C}, {\tt G}\} \cup \{M_1, M_2, \ldots M_k\}$ to denote the alphabet used in composite DNA, including $k$ composite symbols $M_{i,1\le i\le k}$ and four standard letters in $\Sigma_0=\{{\tt A}, {\tt T}, {\tt C}, {\tt G}\}$. We next define two biological constraints in a DNA storage channel: the runlength-limited constraint and the ${\tt G}{\tt C}$-content constraint.
%\vspace{0.1mm}
\subsection{Constrained Coding for Standard DNA Synthesis} 
%\begin{definition}
Consider a DNA sequence $\bx$ consisting of $n$ nucleotides, i.e. $\bx=x_1x_2\ldots x_n$ where $x_i\in \{{\tt A}, {\tt T}, {\tt C}, {\tt G}\}$. Given $\ell>0$, we say $\bx$ is $\ell$-runlength limited (or $\ell$-RLL in short) if any run of the same nucleotide is at most $\ell$. On the other hand, the ${\tt G}{\tt C}$-content of $\bx$ refers to the percentage of nucleotides that are either ${\tt G}$ or ${\tt C}$. Formally, for $a\in \{{\tt A}, {\tt T}, {\tt C}, {\tt G}\}$, we use ${\rm wt}_a(\bx)$ to denote the number of $a$ in $\bx$, i.e. ${\rm wt}_a(\bx)=| \{i\in [1,n]: x_i=a\} |$.  The ${\tt G}{\tt C}$-content of $\bx$, denoted by $\omega(\bx)$, is then computed by $\omega(\bx)=1/n ({\rm wt}_{{\tt G}}(\bx)+{\rm wt}_{{\tt C}}(\bx))$. 
\begin{comment}  
\begin{equation*}
\omega(\bx)=\frac{1}{n} \sum_{i=1}^{n} \varphi(x_i), \text{ where }  \varphi(x_i)= 
\left\{ \begin{array}{rcl}
0& \mbox{for }  x_i \in \{{\tt A}, {\tt T}\} \\
1 & \mbox{for }  x_i \in \{{\tt C}, {\tt G}\}.
\end{array}\right.
%0 \text{ if } x_i \in \{{\tt A}, {\tt T}\}, \varphi(x_i)=1 \text{ if } x_i \in \{{\tt C}, {\tt G}\}. 
\end{equation*}
\end{comment}
%\end{definition}
Given $\epsilon\ge 0$, we say that $\bx$ is $\epsilon$-balanced if the ${\tt G}{\tt C}$-content of $\bx$ satisfies that $|\omega(\bx)-0.5|\le \epsilon$, or $\omega(\bx) \in [0.5-\epsilon, 0.5+\epsilon]$. Equivalently, we have 
\begin{equation}
%|\omega(\bx)-0.5|\le \epsilon, \text{ or } 
n/2-\epsilon n \le {\rm wt}_{{\tt G}}(\bx)+{\rm wt}_{{\tt C}}(\bx) \le n/2+\epsilon n. \label{epsilon-def}
\end{equation}
In particular, when $n$ is even and $\epsilon=0$, $\omega(\bx)=0.5$, we say $\bx$ is ${\tt G}{\tt C}$-balanced. Over a binary alphabet, a sequence $\bu \in \{0, 1\}^n$ is called {\em balanced} if the number of ones in $\bu$ is ${\rm wt}_1(\bu)=n/2$. We observe that if $\bx$ is $\epsilon$-balanced then we also have $n/2-\epsilon n \le {\rm wt}_{{\tt A}}(\bx)+{\rm wt}_{{\tt T}}(\bx) \le n/2+\epsilon n.$
%$|\omega(\bx)-0.5|\le \epsilon$, in other words, we have $\omega(\bx) \in [0.5-\epsilon, 0.5+\epsilon]$. In particular, when $n$ is even and $\epsilon=0$, if $\omega(\bx)=0.5$ we say $\bx$ is ${\tt G}{\tt C}$-balance. Over a binary alphabet, a sequence $\bu \in \{0, 1\}^n$ is called {\em balanced} if the number of ones in $\bu$, or its weight ${\rm wt}(\bu)$, is exactly $n/2$.
\vspace{0.5mm}

Most literature experiments used DNA sequences whose ${\tt G}{\tt C}$-content is close to $50\%$, and the maximum run is of length at most six \cite{erlich:2017,hec:2019, chee:tandem,caikui:ep}.  %For example, the DNA experiment in \cite{erlich:2017} required $\ell=3$ and $\epsilon=0.05$. 
Prior art coding techniques, that simultaneously enforce the RLL constraint and the ${\tt G}{\tt C}$-content constraint for arbitrary values of $\ell$ and $\epsilon$, were proposed in \cite{thanh:dna, xhe:2022}. Due to the unique structure of the alphabet symbols in composite DNA, it is required to develop new coding methods, specifically targeting the composite DNA method. 

\subsection{Constrained Coding for Composite DNA Synthesis} %and Problem Statement} 

%We use $\Sigma_k=\{{\tt A}, {\tt T}, {\tt C}, {\tt G}\} \cup \{M_1, M_2, \ldots M_k\}$ to denote the alphabet used in composite DNA, including $k$ composite symbols $M_{i,1\le i\le k}$ and four standard letters in $\Sigma_0=\{{\tt A}, {\tt T}, {\tt C}, {\tt G}\}$. We next define two biological constraints in a DNA storage channel: the runlength-limited constraint and the ${\tt G}{\tt C}$-content constraint.

%Consider the following example.
%and $\by={\tt A} {\tt C} {\tt M} {\tt G} {\tt G} {\tt M} {\tt T} {\tt A}$  (see Figure~\ref{fig1})
\begin{example}\label{exp1} 
Consider $\ell=3, \epsilon=0.1$, and a data sequence $\bx={\tt A} {\tt C} {\tt M} {\tt C} {\tt M} {\tt A} {\tt T} {\tt A}$ over the alphabet $\Sigma=\{{\tt A}, {\tt T}, {\tt C}, {\tt G}, {\tt M}\}$, where the letter ${\tt M}$ is a mixture of two nucleotides: ${\tt A}$ and ${\tt C}$. Observe that, regardless of the mixing ratio in letter ${\tt M}$, the synthesis machine produces four different sequences from $\bx$ and a huge number of copies for each of these sequences: 
\begin{align*}
\bx^{(1)} &= {\tt A} {\tt C} {\tt A} {\tt C} {\tt A} {\tt A} {\tt T} {\tt A}, \text{  }  \bx^{(2)} = {\tt A} {\tt C} {\tt A} {\tt C} {\tt C} {\tt A} {\tt T} {\tt A}, \\
\bx^{(3)}  &= {\tt A} {\tt C} {\tt C} {\tt C} {\tt A} {\tt A} {\tt T} {\tt A}, \text{  }  \bx^{(4)} = {\tt A} {\color{red}{{\tt C} {\tt C} {\tt C} {\tt C}}} {\tt A} {\tt T} {\tt A}.
\end{align*}
The sequence $\bx^{(4)} $ has a run of length 4 (as highlighted in red). Therefore, $\bx^{(4)} $ is not $3$-RLL. In addition, we also verify that the first sequence $\bx_1$ is not $\epsilon$-balanced for $\epsilon=0.1$. 

On the other hand, consider $\by={\tt A} {\tt C} {\tt M} {\tt G} {\tt G} {\tt M} {\tt T} {\tt A}$. We verify that all four different sequences, that can be obtained from $\by$ during the synthesis process, are both $\ell$-RLL and $\epsilon$-balanced for $\ell=3$ and $\epsilon=0.1$ (see Figure~\ref{fig1}). In this work, we design efficient constrained codes that only accept those data sequences having the same property as sequence $\by$. 
\end{example}

\begin{figure}[h!]
\begin{center}
\includegraphics[width=8.25cm]{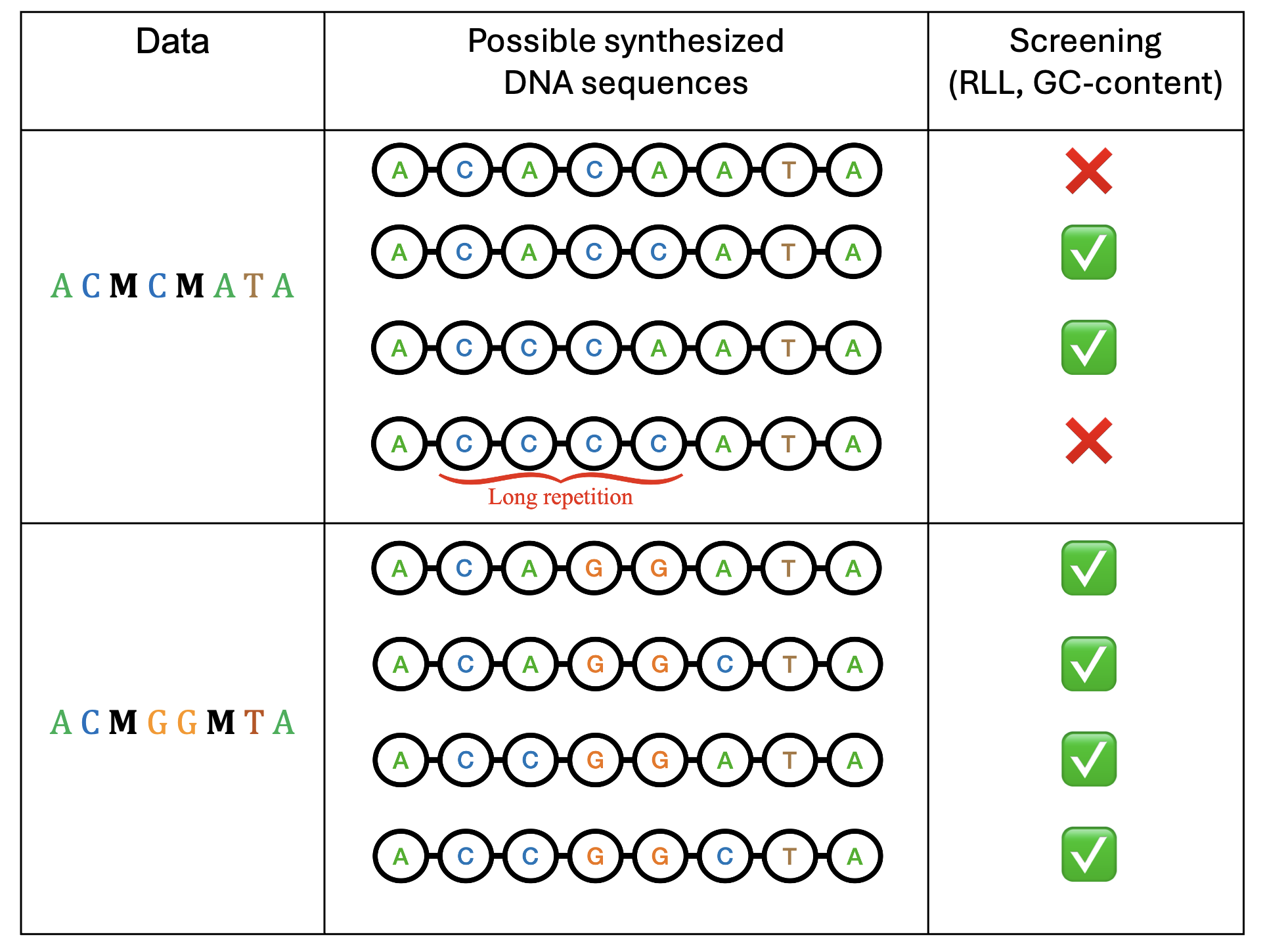}
\end{center}
%\begin{center}
%\includegraphics[width=14cm]{drawDNA1.pdf}
%\end{center}
\caption{Possible synthesized DNA sequences from two data sequences of length $8$ in composite DNA alphabet. Given $\ell=3$ and $\epsilon=0.1$, i.e. the total number of ${\tt G}$ and ${\tt C}$ is within $[3,5]$.} \label{fig1}
\end{figure}

We use $\Sigma_k=\{{\tt A}, {\tt T}, {\tt C}, {\tt G}\} \cup \{M_1, M_2, \ldots M_k\}$ to denote the alphabet used in composite DNA, including four standard letters in $\Sigma_0=\{{\tt A}, {\tt T}, {\tt C}, {\tt G}\}$ and $k$ composite letters $M_{1}, \ldots M_k$. %for $1\le i\le k$. 

\begin{definition} Given an alphabet $\Sigma_k$ and a sequence $\bx \in \Sigma_k^n$, let $\B(\bx; \Sigma_k)$ denote the set of all DNA sequences of length $n$ over the alphabet $\Sigma_0=\{{\tt A}, {\tt T}, {\tt C}, {\tt G}\}$, that can be obtained from $\bx$ during the synthesis process. 
\begin{enumerate}[(i)]
\item Given $\ell>0$, a code $\C\subseteq \Sigma_k^n$ is said to be an $\ell$-RLL code if and only if for all $\bx\in \C$ and all $\by\in \B(\bx; \Sigma_k)$ we have $\by$ is $\ell$-RLL. 
\item Given $\epsilon\in[0,0.5]$, a code $\C\subseteq \Sigma_k^n$ is said to be an $\epsilon$-balanced code if and only if for all $\bx\in \C$ and all $\by\in \B(\bx; \Sigma_k)$, we have $\by$ is $\epsilon$-balanced. 
\item Given $\ell>0, \epsilon\in[0,0.5]$, a code $\C\subseteq \Sigma_k^n$ is said to be an $(\ell,\epsilon)$-constrained code if and only if for all $\bx\in \C$ and all $\by\in \B(\bx; \Sigma_k)$, we have $\by$ is $\ell$-RLL and $\by$ is $\epsilon$-balanced. 
\end{enumerate}
\end{definition}

We use $\mathcal{A}(n,\ell;\Sigma_k)$ to denote the cardinality of the largest $\ell$-RLL code. The channel capacity of the runlength-limited constraint, denoted by ${\bf cap}_{\ell;\Sigma_k}$, is computed as follows: 
\begin{align*}
{\bf cap}_{\ell; \Sigma_k}=\lim_{n\to\infty} \frac{\log_2 \mathcal{A}(n,\ell;\Sigma_k)}{n}. 
\end{align*}
Similarly, we define $\mathcal{A}(n,\epsilon;\Sigma_k)$, $\mathcal{A}(n,\ell,\epsilon;\Sigma_k)$, and compute the channel capacities ${\bf cap}_{\epsilon; \Sigma_k}, {\bf cap}_{\ell,\epsilon; \Sigma_k}$ as follows:
\begin{align*}
{\bf cap}_{\epsilon; \Sigma_k}&=\lim_{n\to\infty} \frac{\log_2 \mathcal{A}(n,\epsilon;\Sigma_k)}{n}, \text{ and } \\
{\bf cap}_{\ell,\epsilon; \Sigma_k}&=\lim_{n\to\infty} \frac{\log_2 \mathcal{A}(n,\ell,\epsilon;\Sigma_k)}{n}.
\end{align*}

The rest of the paper is organized as follows. For each constraint, we compute the capacity of the constrained channel and provide efficient constructions of codes that obey the constraint, particularly focusing on the parameters $\ell, \epsilon,$ and $\Sigma_k$ such that ${\bf cap}_{\ell; \Sigma_k}, {\bf cap}_{\epsilon; \Sigma_k}, {\bf cap}_{\ell,\epsilon; \Sigma_k}>2$.

\section{Coding for RLL Constraint in Composite DNA}\label{RLLsection}

\subsection{Computing the Channel Capacity} 

Given $\ell>0$ and an alphabet $\Sigma_k$, we first construct the {\em forbidden set} $\mathcal{F}(\ell;\Sigma_k)$, a finite set of substrings that must be avoided in all codewords. %We use $\mathcal{F}(\ell;\Sigma_k)$ to denote such a forbidden set. 
\begin{example}\label{sigma1}
Consider an alphabet $\Sigma_1= \{{\tt A}, {\tt T}, {\tt C}, {\tt G}, {\tt M}\}$, where ${\tt M}$ consists of a mixture of ${\tt A}$ and ${\tt C}$, we have: 
\begin{equation*}
\mathcal{F}(\ell,\Sigma_1) = \Big\{ \{{\tt A}, {\tt M}\}^{\ell+1},  \{{\tt C}, {\tt M}\}^{\ell+1}, {\tt T}^{\ell+1}, {\tt G}^{\ell+1} \Big\},
\end{equation*}
and its size is $|\mathcal{F}(\ell,\Sigma_1)|=(2\times 2^{\ell+1}-1)+1+1=2^{\ell+2}+1$. 
\end{example}

Given the forbidden set $\mathcal{F}(\ell;\Sigma_k)$, we view such a $\ell$-RLL code as a constrained code where the codewords can be generated by all paths in a labeled graph $ {\bf {G}} = ({\bf V}, {\bf E}) $, where each note is a substring of length $\ell$, i.e. $ {\bf V}= \Sigma_k^{\ell} $. In addition, there is an edge from a node $\bu$ to a node $\bv$ if and only if $\bu_{[2,\ell]}\equiv \bv_{[1,\ell-1]}$ and the substring $\bw$ of length $(\ell+1)$ where $\bw=u_1 \bu_{[2,\ell]} v_{\ell} \notin \mathcal{F}(\ell;\Sigma_k)$. Let ${\bf D}(\ell;\Sigma_k)$ be the adjacency matrix of ${\bf {G}}$. It is well known that, the capacity of the $\ell$-RLL constraint channel is computed by: 
\begin{equation*}
{\bf cap}_{\ell; \Sigma_k} = \log_2 \lambda_{\ell; \Sigma_k}, 
\end{equation*} 
where $\lambda_{\ell; \Sigma_k}$ is the largest real eigenvalue of the matrix ${\bf D}(\ell;\Sigma_k)$ (refer to \cite{shannon:1948, immink:book}). Therefore, It is possible to compute $\lambda_{\ell; \Sigma_k}$ in every case and thus find the capacity. Several results are computed and tabulated in Table~\ref{table1} and Table~\ref{table1'}. %, particularly the cases, where the channel capacity is strictly larger than 2 bits/symbol, are highlighted in bold.

\begin{example}
Consider $\ell=1$, and $\Sigma_1= \{{\tt A}, {\tt T}, {\tt C}, {\tt G}, {\tt M}\}$, where ${\tt M}$ consists of a mixture of ${\tt A}$ and ${\tt C}$. The forbidden set is then 
\begin{equation*}
\mathcal{F} = \Big\{ {\tt A}^2, {\tt A}{\tt M}, {\tt M}{\tt A}, {\tt M}^2,  {\tt C}^2, {\tt C}{\tt M}, {\tt M}{\tt C}, {\tt T}^{2}, {\tt G}^{2} \Big\}, \text{ and } |\mathcal{F}|=9.
\end{equation*}
We then obtain $ {\bf V} = \{{\tt A}, {\tt T}, {\tt C}, {\tt G}, {\tt M}\}$, and the matrix ${\bf D}$:
\begin{equation*}
{\bf D} = 
\left[
\begin{array}{ccccc} 
0& 1& 1& 1& 0\\
1& 0& 1& 1& 1\\
1& 1& 0& 1& 0\\
1& 1& 1& 0& 1\\
0& 1& 0& 1& 0
\end{array} \right].\label{eq:D_ex}
\end{equation*} 
The largest real eigenvalue of ${\bf D}$ is $\lambda=3.323$, and hence, the channel capacity is $\log_2 \lambda=\log_2 3.323=1.733$. 
\end{example} 

\begin{table}[h]
    \centering
    %\captionsetup{labelsep=colon}  % Customize the caption
    \begin{tabular}{|c|c|c|c|c|c|c|}
        \hline
        $\ell$ & 1 & 2 & 3  & 4 & 5 & 6\\
        \hline
        (i) ${\tt M}: {\tt A}|{\tt C}$   & 1.733 & {\bf 2.170} & {\bf 2.271} & {\bf 2.303} & {\bf 2.315} &{\bf 2.319}\\
        \hline
        (ii) ${\tt M}: {\tt A}|{\tt T}|{\tt C}$ & 1.626 & {\bf 2.121} & {\bf 2.251} & {\bf 2.295} & {\bf 2.311} &{\bf 2.318}\\
        \hline
        (iii) ${\tt M}: {\tt A}|{\tt T}|{\tt C}|{\tt G}$ & 1.585 & {\bf 2.076} & {\bf 2.231} & {\bf 2.287} & {\bf 2.308} & {\bf 2.316}\\
        \hline
    \end{tabular}
    \caption{The channel capacity of the $\ell$-RLL constraint in composite DNA, using an alphabet $\Sigma_1= \{{\tt A}, {\tt T}, {\tt C}, {\tt G}, {\tt M}\}$. %For practical applications, we consider $1\le \ell\le 6$. 
There are three subcases when the composite letter ${\tt M}$ consists of a mixture of (i) any two nucleotides, or (ii) any three nucleotides, or (iii) all four nucleotides. The channel capacity is strictly larger than 2 bits/symbol for all $\ell\ge 2$ (the cases are highlighted in bold).}\label{table1}
\end{table}

\begin{table}[h]
			\centering
			   \begin{tabular}{|c|c|c|c|c|c|c|}
				\hline
				$\ell$ & 1 & 2 & 3 & 4 & 5 & 6 \\ \hline
				${\tt M}: {\tt A}|{\tt T}, {\tt N}: {\tt C}|{\tt G}$ & 1.900 & 2.418 & 2.535 & 2.569 & 2.580 & 2.583 \\ \hline
				${\tt M}: {\tt A}|{\tt T}, {\tt N}: {\tt A}|{\tt G}$ & 1.918 & 2.392 & 2.512 & 2.554 & 2.571 & 2.579 \\ \hline
				${\tt M}: {\tt A}|{\tt T}, {\tt N}: {\tt A}|{\tt C}|{\tt G}$& 1.806 & 2.356 & 2.500 & 2.550 & 2.570 & 2.578 \\ \hline
				${\tt M}: {\tt A}|{\tt T}, {\tt N}: {\tt A}|{\tt T}|{\tt G}$ & 1.821 & 2.331 & 2.478 & 2.536 & 2.562 & 2.574 \\ \hline
				${\tt M}: {\tt A}|{\tt T}|{\tt C}, {\tt N}: {\tt A}|{\tt T}|{\tt G}$ & 1.694 & 2.289 & 2.465 & 2.532 & 2.560 & 2.573 \\ \hline
			\end{tabular}
			 \caption{The channel capacity of the $\ell$-RLL constraint in composite DNA, using an alphabet $\Sigma_2= \{{\tt A}, {\tt T}, {\tt C}, {\tt G}, {\tt M}, {\tt N}\}$.}\label{table1'}
		\end{table}

Given $\ell, \Sigma_k$, the problem of finding an explicit formula for $\mathcal{A}(n,\ell;\Sigma_k)$, i.e., the cardinality of the largest $\ell$-RLL code (or similarly, the problem of finding an explicit formula for the channel capacity ${\bf cap}_{\ell; \Sigma_k}$) is deferred to future work.

\subsection{Efficient Coding Schemes} 
In this section, we show that for certain values of $\ell, \Sigma_k$, a linear-time encoder (and a corresponding linear-time decoder) exists to encode $\ell$-RLL code with only a single redundancy symbol. In such cases, we then have $\mathcal{A}(n,\ell;\Sigma_k)\ge |\Sigma_k|^{n-1}$. Our encoding method is based on the {\em sequence replacement technique}, which has been widely applied in the literature (for example, refer to \cite{srt:1, srt:2, srt:3, srt:4}). Our construction is as follows. 
\vspace{1mm}

\begin{comment}
This is an efficient method for removing forbidden substrings from a source word. %The advantage of this technique is the low complexity of the encoder and decoder, which are also suitable for parallel implementation. 
Over the alphabet $\Sigma_0=\{{\tt A}, {\tt T}, {\tt C}, {\tt G}\}$ in standard DNA synthesis, in our companion work \cite{thanh:dna}, we presented a construction of such an encoder for arbitrary $\ell>0$, and provided an upper bound for codeword's length $n$, so that the redundancy of the encoder is only one symbol (refer to Table III, \cite{thanh:dna}). In this work, we extend the construction in \cite{thanh:dna} to encode $\ell$-RLL codes over the composite DNA alphabet as follows.
\vspace{1mm}
\end{comment}

Given $\ell>0$ and an arbitrary alphabet $\Sigma_k$, we first construct the forbidden set $\mathcal{F}(\ell; \Sigma_k)$. 
\vspace{1mm}

\noindent{\bf $\ell$-RLL Encoder for Composite DNA: $\enc_{\ell; \Sigma_k}$} 
\vspace{1mm}

{\sc Input}: $\bx \in \Sigma_k^{n-1}$\\
{\sc Output}: $\bc=\enc_{\ell; \Sigma_k}(\bx) \in \Sigma_k^{n}$ such that $\bc$ contains no forbidden substring of length $(\ell+1)$ in $\mathcal{F}(\ell,\Sigma_k)$ \\[-2mm]
\vspace{1mm}

\noindent{\em Screening Step}. The encoder first appends `${\tt A}$' to the end of $\bx$, yielding an $n$-symbols sequence, $\by=\bx{\tt A}$. The encoder then checks $\by$, and if there is no forbidden substring of length $(\ell+1)$ in $\by$, the output is simply $\bc=\by=\bx{\tt A}$. Otherwise, it proceeds to the {\em replacement procedure}.
\vspace{1mm}

\noindent{\em Replacement Procedure}. Let the current word $\by=\bu {\bf f} \bv$, where, the prefix $\bu$ has no substring in $\mathcal{F}(\ell,\Sigma_k)$ and ${\bf f}$ is a forbidden substring, ${\bf f}\in \mathcal{F}(\ell,\Sigma_k)$. Suppose that ${\bf f}$ starts at position $p$, where $1\leq p\leq n-\ell$. The encoder removes ${\bf f}$ and updates the current word to be $\by=\bu\bv {\bf R} \alpha$, where the {\em pointer} ${\bf R} \alpha$ is used to represent the position $p$ and the removed substring ${\bf f}$ for unique decoding. In addition, 
\begin{enumerate}[(i)]
\item ${\bf R} \in \Sigma_k^{\ell}$,
\item $\alpha \in \Sigma_k \setminus \{{\tt A}\}$,
\end{enumerate}
Note that the number of unique combinations of the pointer ${\bf R} e$ equals $(\Sigma_k-1)\Sigma_k^{\ell}$. To make the replacement pointer an injective mapping, the following inequality must hold: 
\begin{align}
(n-\ell) |\mathcal{F}(\ell; \Sigma_k)| &\le (|\Sigma_k|-1)\times |\Sigma_k|^\ell, \text{ or } \nonumber \\
 n &\le \frac{(|\Sigma_k|-1)\times |\Sigma_k|^\ell}{|\mathcal{F}(\ell; \Sigma_k)|}+\ell. \label{UBmax}
%\text{ equivalently, we have: } n &\le \frac{4\times 5^\ell}{2^{\ell+2}+1}+\ell. 
\end{align}
The current word $\by=\bu\bv {\bf R} \alpha$ is of length $n$. If, after the replacement, $\by$ contains no forbidden substring, then the encoder outputs $\bc=\by$ as the codeword. Otherwise, the encoder repeats the replacement procedure for the current word $\by$ until all forbidden substrings have been removed. %Note that during every step, the length of the codeword is preserved. %Since the last symbol in any additional pointer is nonzero, the concatenation between any two consecutive pointers ${\bf R_1} e_1 {\bf R_2} e_2$ does not produce any substring $0^{\ell}$, and hence this procedure is guaranteed to terminate. 
\begin{example}[Continuing from Example~\ref{sigma1}]  When $\Sigma_1= \{{\tt A}, {\tt T}, {\tt C}, {\tt G}, {\tt M}\}$, where ${\tt M}$ consists of a mixture of ${\tt A}$ and ${\tt C}$, we have $|\mathcal{F}(\ell,\Sigma_1)|=2^{\ell+2}+1$. Our encoder works for 
\begin{equation*}
n\le \frac{4\times 5^\ell}{2^{\ell+2}+1}+\ell.
\end{equation*} 
Particularly, when $\ell=6$, the bound implies that our encoder uses only one redundant symbol for all $n\le250$. More results are computed and tabulated in Table~\ref{compare}.
\end{example}

%\begin{remark}
\begin{theorem}\label{encoder:defined}
Our encoder $\enc_{\ell; \Sigma_k}$ is well-defined. %all we need to prove is its convergence.
In other words, the replacement procedure is guaranteed to terminate. 
\end{theorem}
\begin{proof} 
We model the encoder as a walk on a directed graph, where nodes represent sequences of length $n$ and edges represent the replacement routine, i.e. there is an edge from a sequence $\by_i$ to a sequence $\by_{j}$ if $\by_i$ has a forbidden substring and a replacement is performed, resulting a sequence $\by_{j}$. It is easy to see that the {\em out-degree} of all nodes is at most one (as the replacement routine is injective). In addition, if there is an edge from a sequence $\by_i$ to a sequence $\by_{j}$, then from $\by_j$, it is uniquely decodable to obtain $\by_i$. In other words, the {\em in-degree} of all nodes is also at most one. The number of nodes is finite (bounded above by $\Sigma_k^n$). Suppose that the replacement procedure does not terminate, i.e., a {\em loop} (or {\em cycle}) exists. Suppose from a data $\bx$, we have a walk $\by=\bx{\tt A} \to \by_1 \to ... \to (\by_i \to \by_{i+1} \to \ldots \to \by_j \to \by_i)$. Since $\by$ ends with ${\tt A}$ while all other nodes end with a different symbol (according to the replacement procedure), we must have $j> i \ge 1$, and hence, we assume that all nodes in $\bS=\{\by, \by_1, \ldots \by_{i-1}\}$ do not belong to the loop. Such a set $\bS$ is a non-empty set. In a special case when $i=1$, then $\by_{i-1}\equiv \by$ and $\bS=\{\by\}$. We illustrate the scenario as follows. 
\begin{figure}[h!]
\begin{center}
\includegraphics[width=7.5cm]{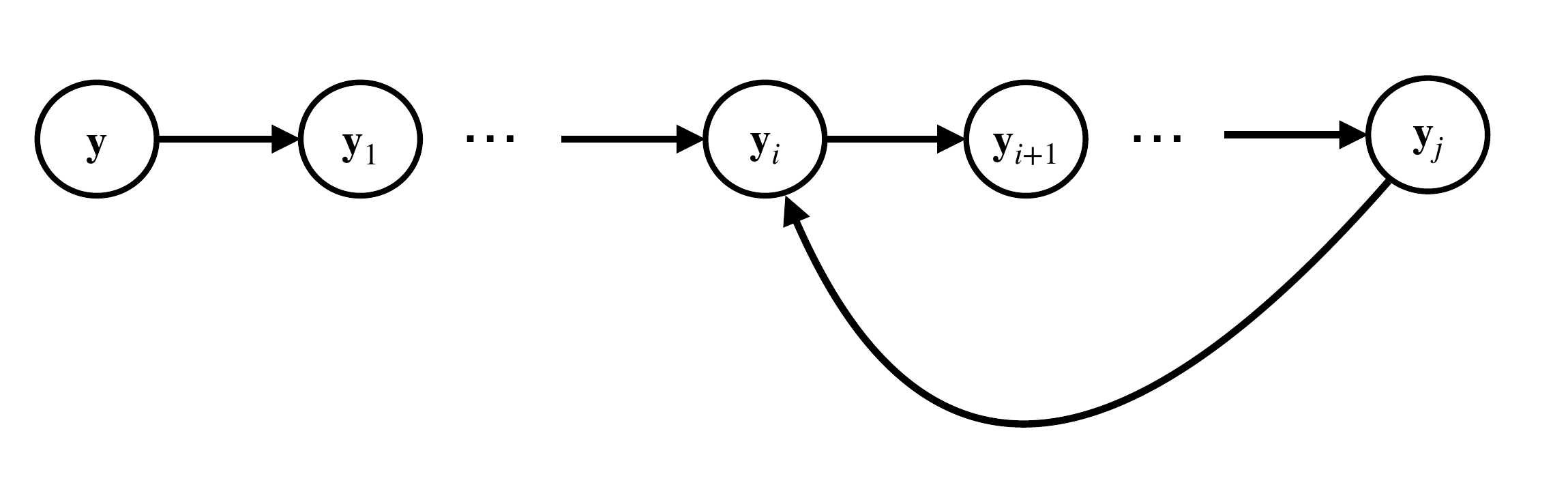}
\end{center}
\end{figure}

\noindent We then have the in-degree of $\by_i$ is at least two: $\by_{i-1} \to \by_i$ and $\by_{j}\to \by_i$, where $j>i>i-1$. Consequently, from $\by_i$, the decoder obtains two different sequences $\by_j$ and $\by_{i-1}$. We have a contradiction. %In conclusion, the replacement procedure is guaranteed to terminate.
%\end{remark}
\end{proof}

%We obtain a similar result (presented in Lemma 1, \cite{srt:3}).  

\begin{lemma}
The average number of replacement iterations done by the $\ell$-RLL encoder $\enc_{\ell; \Sigma_k}$ is at most $|\Sigma_k| = O(1)$.
\end{lemma}
\begin{proof}
As discussed in Remark~\ref{encoder:defined}, the in-degree (and out-degree) of all nodes in such a graph is at most one. Consequently, given two distinct data sequences $\bx_1$ and $\bx_2$, two paths generated by $\bx_1$ and $\bx_2$ are disjoint. The sum of replacement iterations the encoder does for all data sequences is bounded above by $|\Sigma_k|^n$, and hence, the average number of replacement iterations is at most $|\Sigma_k|^n /| \Sigma_k|^{n-1} = |\Sigma_k|=O(1)$. %, for constant $k$.  
\end{proof}

The corresponding decoder $\dec_{\ell; \Sigma_k}$ is as follows.  %as follows. 
\vspace{1mm}

\noindent{\bf $\ell$-RLL Decoder for Composite DNA: $\dec_{\ell; \Sigma_k}$} 

{\sc Input}: $\bc \in \Sigma_k^{n}$,  $\bc=\enc_{\ell; \Sigma_k}(\bx)$ for some unknown $\bx$ \\
{\sc Output}: $\bx=\dec_{\ell; \Sigma_k}(\bc) \in \Sigma_k^{n-1}$ \\[-2mm]
%\vspace{1mm}

\noindent{\em Decoding Procedure}. The decoder checks from the right to the left. If the last symbol is `${\tt A}$', the decoder simply removes the last symbol and identifies the first $n-1$ symbols are source data $\bx$. On the other hand, if the last symbol is not `${\tt A}$', the decoder takes the suffix of length $(\ell+1)$, identifies it as the pointer ${\bf R} \alpha$ for some $\alpha\neq {\tt A}$ and ${\bf R}\in\Sigma_k^{\ell}$, representing a unique substring ${\bf f}\in\mathcal{F}(\ell,\Sigma_k)$. The decoder then adds back the substring ${\bf f}$ accordingly. It terminates when the last symbol of the sequence is `${\tt A}$'.

\begin{remark}\label{remark:extend}
Suppose the codewords' length is larger than the proposed upper bound in \eqref{UBmax}. In that case, one can use the concatenation coding method to divide the data into subwords whose length is less than the upper bound in \eqref{UBmax}. The subwords are encoded in parallel and then concatenated to form a codeword. %We defer the details of such an extended encoder to the full
%version of this paper. 
%In addition, this algorithm can be
%%easily extended for the case of arbitrary length n  N. The
%main idea is that we divide the source data into subwords of
%length N −1, encode separately each subword and concatenate
%them. The representation pointer needs to be modified so
%that the concatenation between any two encoded subwords does not contain a forbidden substring. To do so, 

\end{remark}

%\begin{theorem}
%Let $N=(q-2)q^{\ell-1}+\ell$. There exists a pair of algorithms $\enc_{\rm RLL}: \Sigma_q^{N-1} \to \Sigma_q^N$ and $\dec_{\rm RLL}: \Sigma_q^N \to  \Sigma_q^{N-1}$ such that for all $\bx\in\Sigma_q^{N-1}$, $\enc_{\rm RLL}(\bx)$ is $\ell$-runlength limited and $\dec_{\rm RLL} \circ \enc_{\rm RLL} (\bx)=\bx$. In addition, for all $n=mN$ for some $m\ge1$, $\bx=\bx_1\bx_2\ldots\bx_m$ where $\bx_i\in\Sigma_q^{N-1}$, we define $\enc_{\rm RLL}^B(\bx)\triangleq \enc_{\rm RLL} (\bx_1) \enc_{\rm RLL} (\bx_2) \ldots \enc_{\rm RLL} (\bx_m)$ then $\enc_{\rm RLL}^B(\bx)$ is $\ell$-runlength limited.
%\end{theorem} 

\begin{table}[h!]
\centering 
 %\begin{tabular}{|c|| c| c|}
 \begin{tabular}{ |P{0.5cm}|| P{3cm}|P{3cm}|}
 \hline
 $\ell$ & $\Sigma_1=\{{\tt A}, {\tt T}, {\tt C}, {\tt G}, {\tt M_1}\}$, where ${\tt M_1}: {\tt A}|{\tt C}$ &  $\Sigma_2=\{{\tt A}, {\tt T}, {\tt C}, {\tt G}, {\tt M_1}, {\tt M_2}\}$, where ${\tt M_1}: {\tt A}|{\tt T}, {\tt M_2}: {\tt C}|{\tt G}$\\[1ex]
 \hline
 3  &   $n\le 19$ & $n\le 20$\\
 4  &   $n\le 43$ &  $n\le55$\\
 5  &    $n\le102$ &  $n\le158$ \\
 6 & $n\le250$ &  $n\le463$ \\
\hline
\end{tabular}
%\caption{Maximum length $n$ that our proposed encoder uses only one redundant symbol.}
\caption{$\ell$-RLL Encoder with One Redundant Symbol.}
\label{compare}
\end{table}

%~\backslash~n_{\rm max}$

\section{Coding for ${\tt G}{\tt C}$-Content Constraint} %in Composite DNA}

\subsection{Channel Capacity and Efficient Encoding Schemes}\label{epsilon-subsection}

From Section III, we observe that for constant $\ell>0$, ${\bf cap}_{\ell; \Sigma_k} < \log_2 |\Sigma_k|$. In this section, we show that for several values of $\epsilon$ and $\Sigma_k$, there is no rate loss, i.e. the channel capacity is ${\bf cap}_{\epsilon; \Sigma_k}=\log_2 |\Sigma_k|$.
%\vspace{1mm}

\begin{definition}
For a positive integer $N$, the {\em DNA-representation} of $N$ is the replacement of symbols in the quaternary representation of $N$ as follows:  $0  \leftrightarrow {\bf A}, 1  \leftrightarrow {\bf T}, 2  \leftrightarrow {\bf C}, \text{ and } 3  \leftrightarrow {\bf G}.$
\end{definition} 

\begin{example}
If $N=100$, the quaternary representation of length 4 of $N$ is $1210$, hence, the DNA-representation of $N$ is ${\bf T}{\bf C}{\bf T}{\bf A}$. Similarly, when $N=55$, the quaternary representation of length 4 of $N$ is $0313$, thus the DNA-representation of $N$ is ${\bf A}{\bf G}{\bf T}{\bf G}$.
\end{example}

%\begin{definition}
\begin{comment}
Given a composite alphabet $\Sigma_k$ and a sequence $\bx\in \Sigma_k^n$, the {\em DNA projection} of $\bx$, denoted by ${\rm P}(\bx)$, is the longest subsequence of $\bx$ that only consists of symbols in $\Sigma_0=\{{\tt A}, {\tt T}, {\tt C}, {\tt G}\}$.  
%\end{definition}
For example, consider $\Sigma_1=\{{\tt A}, {\tt T}, {\tt C}, {\tt G}, {\tt M}\}$, and $\bx= {\tt A}{\tt M}{\tt G}{\tt C}{\tt M}$, $\by= {\tt M}{\tt M}{\tt T}{\tt C}{\tt M}$. We then obtain ${\rm P}(\bx)={\tt A}{\tt G}{\tt C}$ and ${\rm P}(\by)={\tt T}{\tt C}$. 
\end{comment}

\begin{theorem}\label{alphabet5} 
Consider an arbitrary composite alphabet of size 5: $\Sigma_1=\{{\tt A}, {\tt T}, {\tt C}, {\tt G}, {\tt M}\}$, where ${\tt M}$ can be any mixture of the four DNA nucleotides. For any $\epsilon\ge 0.1$, the channel capacity is ${\bf cap}_{\epsilon; \Sigma_1}=\log_2 5$. In addition, there exists a linear-time encoder (and a corresponding decoder) to encode $\epsilon$-balanced code in $\Sigma_1^n$ with redundancy at most $O(\log_4 n)$ symbols. %Moreover, when $\epsilon>0.1$, the redundancy can be further reduced to $O(1)$ symbols. 
\end{theorem}
%\vspace{1mm}

\begin{proof}
Consider the set $\bS$ consisting of all sequences $\bx\in \Sigma_1^n$ such that:
\begin{align}
 &{\rm wt}_{{\tt M}}(\bx) \le 0.2n \le 2\epsilon n \text{ for all } \epsilon\ge 0.1, \text{ and } \label{parity-constraint} \\
&{\rm wt}_{\tt A}(\bx)+{\rm wt}_{\tt T}(\bx) =  {\rm wt}_{\tt C}(\bx)+{\rm wt}_{\tt G}(\bx) = \frac{n-{\rm wt}_{{\tt M}}(\bx)}{2} \label{knuth-constraint}. 
\end{align}
We now show that $\bS$ is an $\epsilon$-balanced code. Consider an arbitrary sequence $\bx\in\bS$. For all $\by\in\B(\bx;\Sigma_1)$, i.e. $\by\in\Sigma_0^n$ can be obtained from $\bx$ during the synthesis, we have 
\begin{align*}
{\rm wt}_{\tt C}(\by)+{\rm wt}_{\tt G}(\by) &\ge {\rm wt}_{\tt C}(\bx)+{\rm wt}_{\tt G}(\bx) \ge 0.5n-0.1n =0.4n,\\ %\ge (0.5-\epsilon)n, \\
{\rm wt}_{\tt C}(\by)+{\rm wt}_{\tt G}(\by) &\le {\rm wt}_{\tt C}(\bx)+{\rm wt}_{\tt G}(\bx)+ {\rm wt}_{\tt M}(\bx) \\
&= \frac{n+{\rm wt}_{\tt M}(\bx)}{2} \\
&\le 0.5n+0.1n = 0.6n.   
\end{align*}
Therefore, for all $\epsilon \ge 0.1$, $\by$ is $\epsilon$-balanced regardless the choice of composite letter ${\tt M}$. Thus, $\bS$ is an $\epsilon$-balanced code. 

It suffices to show that there exists an efficient encoder for codewords in $\bS$ with redundancy at most $O(\log_4 n)$ symbols. Consequently, it implies that the channel capacity is ${\bf cap}_{\epsilon; \Sigma_1}=\log_2 5$. While one parity symbol is sufficient to enforce the constraint \eqref{parity-constraint}, we need $O(\log_4 n)$ symbols to enforce the constraint \eqref{knuth-constraint}, by using the Knuth's balancing technique \cite{alon:1988, knuth:1986}. We next define the one-to-one flipping rule $f$: ${\tt A} \to {\tt C}, {\tt C} \to {\tt A}, {\tt T} \to {\tt G}, {\tt G} \to {\tt T}.$
%\begin{equation*}
%f: {\tt A} \to {\tt C}, {\tt C} \to {\tt A}, {\tt T} \to {\tt G}, {\tt G} \to {\tt T}.
%\end{equation*}
For a sequence $\bx\in\Sigma_1^n$ and $1\le t\le n$, let $f_t(\bx)$ denote the sequence obtained by flipping the first $t$ symbols in $\bx$ according to $f$. For example, if $\bx= {\tt M}{\tt C}{\tt M}{\tt G}{\tt A}{\tt T}$ then $f_4(\bx)={\color{blue}{{\tt M}{\tt A}{\tt M}{\tt T}}}{\tt A}{\tt T}$ and $f_5(\bx)={\color{blue}{{\tt M}{\tt A}{\tt M}{\tt T}{\tt C}}}{\tt T}$. Our encoder is constructed as follows. 
\vspace{1mm}

\noindent{\bf $\epsilon$-balanced Encoder for Composite DNA: $\enc_{\epsilon; \Sigma_1}$} 

{\sc Input}: $\bx \in \Sigma_1^{n-2\ceil{\log_4 n}-1}$, and $n\ge 16$\\
{\sc Output}: $\bc=\enc_{\epsilon; \Sigma_1}(\bx) \in \bS$ \\[-2mm]

\noindent {\em Enforcing \eqref{parity-constraint}}. Select $\alpha\in\Sigma_1$ be the symbol with the minimum weight in $\bx$, i.e ${\rm wt}_{\alpha}(\bx) \le {\rm wt}_{\beta}(\bx)$ for all $\beta\in\{{\tt A}, {\tt T}, {\tt C}, {\tt G}, {\tt M}\}$. Let $\by$ be the sequence obtained from $\bx$ by replacing all ${\tt M}$ with $\alpha$ and $\alpha$ with ${\tt M}$. Let $\bz=\by\alpha$, i.e. appends $\alpha$ to the end of $\by$ (for decoding purpose). We observe that 
\begin{equation*}
{\rm wt}_{{\tt M}}(\bx) \le 1/5 (n-2\ceil{\log_4 n}-1) + 1 \le 0.2 n \text{ for all } n\ge16.
\end{equation*}
\noindent {\em Enforcing \eqref{knuth-constraint}}. The sequence $\bz$ is of length $m=n-2\ceil{\log_4 n}$. For simplicity, we further assume that $m$ is even. In the special case that ${\rm wt}_{\tt A}(\bz)+{\rm wt}_{\tt T}(\bz) =  {\rm wt}_{\tt C}(\bz)+{\rm wt}_{\tt G}(\bz) = 1/2 (m-{\rm wt}_{{\tt M}}(\bz))$, no flipping action is needed. The encoder set the {\em flipping index} $t=0$. The DNA-representation of length $\ceil{\log_4 n}$ of $t=0$ is then ${\tt A}^{\ceil{\log_4 n}}$, and its complement sequence is ${\tt C}^{\ceil{\log_4 n}}$. The encoder then outputs $\bc=\bz({\tt A}{\tt C})^{\ceil{\log_4 n}}$ of length $n$. Clearly, in such a sequence $\bc$, we have 
\begin{equation*}
{\rm wt}_{\tt A}(\bc)+{\rm wt}_{\tt T}(\bc) =  {\rm wt}_{\tt C}(\bc)+{\rm wt}_{\tt G}(\bc) = \frac{n-{\rm wt}_{{\tt M}}(\bc)}{2}.
\end{equation*}
On the other hand, if ${\rm wt}_{\tt A}(\bz)+{\rm wt}_{\tt T}(\bz)< 1/2 (m-{\rm wt}_{{\tt M}}(\bz))$, and ${\rm wt}_{\tt C}(\bz)+{\rm wt}_{\tt G}(\bz) > 1/2 (m-{\rm wt}_{{\tt M}}(\bz))$. If we flip all symbols in $\bz$ then ${\rm wt}_{\tt A}(f_m(\bz))+{\rm wt}_{\tt T}(f_m(\bz))> 1/2 (m-{\rm wt}_{{\tt M}}(\bz))$. Therefore, there must be an index $0<t< m$ so that 
\begin{align*}
{\rm wt}_{\tt A}(f_t(\bz))+{\rm wt}_{\tt T}(f_t(\bz)) &=  {\rm wt}_{\tt C}(f_t(\bz))+{\rm wt}_{\tt G}(f_t(\bz)) \\
&= \frac{m-{\rm wt}_{{\tt M}}(f_t(\bz))}{2}.
\end{align*}
The encoder computes $\bu$, which is the DNA-representation of length $\ceil{\log_4 n}$ of $t$, and sets $\bv$ to be the complement sequence of $\bu$. The encoder then output $\bc=f_t(\bz) (\bu||\bv)$ of length $n$. Similarly, we can show that $\bc$ satisfies the constraint \eqref{knuth-constraint}. It is easy to verify that from a codeword $\bc$, we can decode uniquely the data sequence $\bx\in \Sigma_1^{n-2\ceil{\log_4 n}-1}$ according the suffix of length $2\ceil{\log_4 n}$ of $\bc$ and the last symbol in $\bz$.  
\end{proof}

\begin{lemma}
%Consider $\Sigma_1=\{{\tt A}, {\tt T}, {\tt C}, {\tt G}, {\tt M}\}$, where ${\tt M}$ can be any mixture of the four DNA nucleotides. 
When $\epsilon>0.1$, the redundancy can be further reduced to only $O(1)$ symbols. On the other hand, when $\epsilon<0.1$, the channel capacity ${\bf cap}_{\epsilon; \Sigma_1}<\log_2 5$.
\end{lemma}

\begin{lemma}
Consider an alphabet $\Sigma_1=\{{\tt A}, {\tt T}, {\tt C}, {\tt G}, {\tt M}\}$, where ${\tt M}$ is either a mixture of two nucleotides ${\tt A}, {\tt T}$ or a mixture of two nucleotides ${\tt C}, {\tt G}$. For any $\epsilon\ge 0.1$, there exists a linear-time encoder (and a corresponding linear-time decoder) to encode $\epsilon$-balanced code in $\Sigma_1^n$ with redundancy at most $O(1)$ symbols.  
\end{lemma}

\begin{remark}
When recording the flipping index $t$, using the interleaving sequence $\bu||\bv$ helps to avoid a long run of repeated symbols (see Subsection~\ref{subsection:both}). According to Theorem~\ref{alphabet5}, one can obtain capacity-approaching codes for DNA composite, that the ${\tt G}{\tt C}$-content of all output DNA sequences is always within $(40\%,60\%)$. In addition, the result in Theorem~\ref{alphabet5} can be extended to an arbitrarily larger alphabet (see Corollary~\ref{alphabet6}). %The following result is then immediate.  
\end{remark}

\begin{comment}
\begin{corollary}
Consider an alphabet $\Sigma_1=\{{\tt A}, {\tt T}, {\tt C}, {\tt G}, {\tt M}\}$, where ${\tt M}$ is either a mixture of two nucleotides ${\tt A}, {\tt T}$ or a mixture of two nucleotides ${\tt C}, {\tt G}$. For any $\epsilon\ge 0.1$, there exists a linear-time encoder (and a corresponding linear-time decoder) to encode $\epsilon$-balanced code in $\Sigma_1^n$ with redundancy at most $O(1)$ symbols.  
\end{corollary}
\end{comment}
\begin{corollary}\label{alphabet6} 
Consider an arbitrary composite alphabet: $\Sigma_2=\{{\tt A}, {\tt T}, {\tt C}, {\tt G}, {\tt M}, {\tt N}\}$, where ${\tt M}, {\tt N}$ can be any mixture of the four DNA nucleotides. For any $\epsilon\ge 1/6$, we have ${\bf cap}_{\epsilon; \Sigma_2}=\log_2 6$. In addition, there exists a linear-time encoder (and a corresponding linear-time decoder) to encode $\epsilon$-balanced code in $\Sigma_1^n$ with redundancy at most $O(\log_4 n)$ symbols. %Moreover, when $\epsilon>0.1$, the redundancy can be further reduced to $O(1)$ symbols. 
\end{corollary}

For a composite alphabet of size six, the following result provides the best options for the composite letters to obtain capacity-approaching codes for any $\epsilon\ge 0$. 

\begin{theorem}\label{special-symbol:alphabet6}
Consider an alphabet $\Sigma_2=\{{\tt A}, {\tt T}, {\tt C}, {\tt G}, {\tt M}, {\tt N}\}$, where ${\tt M}$ is a mixture of two nucleotides ${\tt A}, {\tt T}$, and ${\tt N}$ is a mixture of two nucleotides ${\tt C}, {\tt G}$. For any $\epsilon\ge 0$, the channel capacity is ${\bf cap}_{\epsilon; \Sigma_2}=\log_2 6$. In addition, there exists a linear-time encoder (and a corresponding linear-time decoder) to encode $\epsilon$-balanced code in $\Sigma_2^n$ with redundancy at most $O(\log_4 n)$ symbols. When $\epsilon>0, \epsilon n>1$, the redundancy can be further reduced to $O(1)$ symbols. Moreover, for sufficiently large $n$, there exists a linear-time encoder (and a corresponding linear-time decoder) that uses only one redundant symbol. 
\end{theorem}

\begin{proof}
We briefly describe the main idea of our construction and defer the detailed proof to the full paper. Similar to the construction of the encoder $\enc_{\epsilon; \Sigma_1}$ in Theorem~\ref{alphabet5}, we can obtain a linear-time encoder to encode $\epsilon$-balanced code in $\Sigma_2^n$, denoted by $\enc_{\epsilon; \Sigma_2}$, with redundancy $2\ceil{\log_4 n}$ symbols for any $\epsilon\ge0$. Here, we modify the flipping rule $f$ as follows: ${\tt A} \to {\tt C}, {\tt C} \to {\tt A}, {\tt T} \to {\tt G}, {\tt G} \to {\tt T},  {\tt M} \to {\tt N},$ and ${\tt N} \to {\tt M}$. Again, $\ceil{\log_4 n}$ symbols are required to represent the flipping index $t\le n$, and the other $\ceil{\log_4 n}$ symbols represent the complement sequence (to ensure the generating suffix is ${\tt G}{\tt C}$-balanced).
\vspace{1mm} 

In the case that $\epsilon>0$ and $\epsilon n>1$, in our companion work \cite{thanh:dna}, we show that such a flipping index can be found in a set $\bL\subset\{0,1,\ldots n\}$ of constant size:
\begin{equation*}
\bL=\{0, 2\floor{\epsilon n}, 4\floor{\epsilon n}, 6\floor{\epsilon n}\ldots \}, \text{ here } |\bL| \le \floor{1/2\epsilon}+1. 
\end{equation*} 
Consequently, to record the flipping index $t$, we only require $O(1)$ redundant symbols. Finally, we show that the redundancy can be further reduced to only one symbol for sufficiently large $n$. We define the following one-to-one correspondence $\phi$:
\begin{align*}
&{\tt A} \leftrightarrow 00,\quad {\tt T} \leftrightarrow 01,\quad {\tt M}({\tt A}|{\tt T})   \leftrightarrow 02, \\ 
&{\tt C} \leftrightarrow 10,\quad {\tt G} \leftrightarrow 11,\quad {\tt N}({\tt C}|{\tt G})  \leftrightarrow 12.
\end{align*}
Therefore, given a composite sequence $\bx\in\Sigma_2^n$, we have a corresponding sequence $\phi(\bx)$ of length ${2n}$. Let ${\rm Odd}(\phi(\bx))$ be the subsequence $\by$ of $\bx$ consisting of all the symbols at the odd indices of $\phi(\bx)$. We have $\by\in\{0,1\}^n$. It is easy to verify that it suffices to enforce ${\rm wt}(\by)\in [(0.5-\epsilon)n,(0.5+\epsilon)n]$. In our companion work \cite{TT:2020}, we show that one redundant symbol is sufficient for any $n\ge(1/\epsilon^2) \ln n$ (refer to Lemma~\ref{theoremSRT}). 
\end{proof}
%and we set $\bU_\bsg=x_1x_3\cdots x_{2n-1}$ and $\bL_\bsg=x_2x_4\cdots x_{2n}$.
%In other words, $\bsg=\Psi^{-1}(\bU_\bsg || \bL_\bsg)$.
%We refer to $\bU_{\sigma}$ and $\bL_\bsg$ as the {\em upper sequence} and {\em lower sequence} of $\bsg$, respectively. The following result is immediate.

%\begin{lemma}
%Let $\bsg \in \Sigma_4^n$. We have $\bsg$ is $\epsilon$-balanced if and only if $\bU_{\sigma}$ is $\epsilon$-balanced.
%\end{lemma}
\begin{lemma}[Followed by Theorem 3, Nguyen \et{} \cite{TT:2020}]\label{theoremSRT}
Given $\epsilon>0$ and an integer $n$ such that $(1/\epsilon^2) \ln n \le n$, there exists a linear-time encoder $\enc_{\epsilon} :\{0,1\}^{n-1} \to\{0,1\}^n$ such that for all $\bx\in \{0,1\}^{n-1}$ if $\by=\enc_\epsilon(\bx)$ then ${\rm wt}(\by) \in [(0.5-\epsilon)n,(0.5+\epsilon)n]$.
\end{lemma}

\subsection{Enforcing Simultaneously RLL Constraint and ${\tt G}{\tt C}$-Content Constraint in Composite DNA}\label{subsection:both}
The following result is immediate. 
\begin{corollary} We have the following statements. 
\begin{itemize}
\item Consider an arbitrary composite alphabet of size 5: $\Sigma_1=\{{\tt A}, {\tt T}, {\tt C}, {\tt G}, {\tt M}\}$, where ${\tt M}$ can be any mixture of the four DNA nucleotides. For any $\ell>0$ and $\epsilon\ge 0.1$, we have ${\bf cap}_{\ell,\epsilon; \Sigma_1}={\bf cap}_{\ell;\Sigma_1}$.
\item Consider an arbitrary composite alphabet: $\Sigma_2=\{{\tt A}, {\tt T}, {\tt C}, {\tt G}, {\tt M}, {\tt N}\}$, where ${\tt M}, {\tt N}$ can be any mixture of the four DNA nucleotides. For any $\ell>0$ and $\epsilon\ge 1/6$, we have ${\bf cap}_{\ell,\epsilon; \Sigma_2}={\bf cap}_{\ell;\Sigma_2}$.
\item Consider an alphabet $\Sigma_2'=\{{\tt A}, {\tt T}, {\tt C}, {\tt G}, {\tt M}, {\tt N}\}$, where ${\tt M}$ is a mixture of two nucleotides ${\tt A}, {\tt T}$, and ${\tt N}$ is a mixture of two nucleotides ${\tt C}, {\tt G}$. For any $\ell>0$ and $\epsilon\ge0$, we have ${\bf cap}_{\ell,\epsilon;\Sigma_2'}={\bf cap}_{\ell;\Sigma_2'}$.
\end{itemize}
\end{corollary}
Via simple modifications, our proposed encoders, presented in Section~\ref{RLLsection} and Subsection~\ref{epsilon-subsection}, can be efficiently combined to simultaneously enforce both RLL constraint and ${\tt G}{\tt C}$-content constraint in composite DNA. 
\vspace{1mm}

For example, consider a composite alphabet $\Sigma_2=\{{\tt A}, {\tt T}, {\tt C}, {\tt G}, {\tt M}, {\tt N}\}$, where ${\tt M}$ is a mixture of ${\tt A}$ and ${\tt T}$ while ${\tt N}$ is a mixture of ${\tt C}$ and ${\tt G}$. 
%The rest of the section presents an efficient encoder for $(\ell, \epsilon)$-constrained codes over $\Sigma_2^n$ for arbitrary values of $\ell$ and $\epsilon$. 
Suppose that ${\bf r}_\epsilon$ is the redundancy of the encoder in Theorem~\ref{special-symbol:alphabet6} for codewords of length $n$. According to Theorem~\ref{special-symbol:alphabet6}, ${\bf r}_\epsilon=O(\log_4 n)$ when $\epsilon=0$ while ${\bf r}_\epsilon=O(1)$ whenever $\epsilon>0$. On the other hand, suppose that ${\bf r}_\ell$ is the redundancy of the encoder $\enc_{\ell;\Sigma_2}$ from Theorem~\ref{encoder:defined} for $\ell$-RLL codes of length $m$. Recall that ${\bf r}_\ell=1$ if $m$ satisfies the bound in \eqref{UBmax}. It is easy to show that there exists a linear-time encoder for $(\ell, \epsilon)$-constrained codes over $\Sigma_2^n$ with at most ${\bf r}_\ell+{\bf r}_\epsilon+6$ symbols. 
\begin{comment}
\noindent{\bf $(\ell,\epsilon)$-Constrained Encoder for Composite DNA: $\enc_{\ell,\epsilon; \Sigma_2}$} 
{\sc Input}: $\bx \in \Sigma_2^{n-{\bf r}_\ell-{\bf r}_\epsilon-6}$, and $n\ge 16$\\
{\sc Output}: $\bc=\enc_{\ell,\epsilon; \Sigma_2}(\bx)$ such that for all sequences $\by\in \B(\bc; \Sigma_2)$, we have $\by$ is $\ell$-RLL and $\epsilon$-balanced\\[-2mm]
\end{comment}
\vspace{1mm}

We illustrate the construction of $(\ell,\epsilon)$-constrained encoder through the following example, and defer the detailed constructions to the full paper. 

\begin{example}\label{ex:higher}
Suppose $n=300, \epsilon=0.1$, and $\ell=6$. According to Table~\ref{compare}, we have ${\bf r}_\ell=1$. On the other hand, a flipping index can be found in a set $\bL=\{0, 2\floor{\epsilon n}, 4\floor{\epsilon n}, 6\floor{\epsilon n} \ldots \}=\{0,60,120,180,240,300\}$ of size six. Hence, the DNA-representation $\bu$ of the flipping index is of length 2, and the interleaving sequence of $\bu$ and its complement sequence $\bv$ is of length 4. Thus, we have ${\bf r}_{\epsilon}=4$. We now show that there exists an efficient $(\ell,\epsilon)$-constrained encoder with at most ${\bf r}_\ell+{\bf r}_{\epsilon}+6=11$ redundant symbols. 
\vspace{0.1mm}

From the data sequence $\bx\in\Sigma_2^{289}$, we obtain $\by=\enc_{\ell;\Sigma_2}(\bx)\in\Sigma_2^{290}$. We then search for the $\epsilon$-balanced index $t$ in the set $\bL$. Suppose that $\by=y_1y_2\ldots y_t y_{t+1}\ldots y_{290}$. Using the flipping rule $f$ as in Theorem~\ref{special-symbol:alphabet6}: ${\tt A} \to {\tt C}, {\tt C} \to {\tt A}, {\tt T} \to {\tt G}, {\tt G} \to {\tt T},  {\tt M} \to {\tt N},$ and ${\tt N} \to {\tt M}$, we obtain a $\epsilon$-balanced sequence 
\begin{equation*}
\bz={\color{blue}{\underbrace{f(y_1)f(y_2)\ldots f(y_t)}_{\bz_1}}} {\color{black}{ \underbrace{y_{t+1}\ldots y_{290}}_{\bz_2}}} {\color{red}{(\underbrace{\bu||\bv}_{\bz_3})}},
\end{equation*}
\noindent and there is no forbidden subsequence in each $\bz_1, \bz_2$ and $\bz_3$. The 6 additional redundant symbols are to ensure that there is no forbidden subsequence created in the boundaries between $\bz_1$ and $\bz_2$, or between $\bz_2$ and $\bz_3$. We also require these 6 symbols to be ${\tt G}{\tt C}$-balanced. It is easy to see that it suffices to uses 4 symbols between $\bz_1$ and $\bz_2$, and 2 symbols between $\bz_2$ and $\bz_3$. Thus, the encoder uses 11 redundant symbols to encode codewords of length 300, and hence, the rate of the encoder is $\log_2{6^{289}}/300=2.490$ bits/symbol. 
\end{example}

\section{Conclusion}

 We have proposed efficient constrained coding techniques for composite DNA to enforce the biological constraints, including the runlength-limited constraint and the ${\tt G}{\tt C}$-content constraint, into every DNA synthesized oligo, regardless of the mixture of bases in each composite letter and their corresponding ratio. To the best of our knowledge, no such codebooks are known prior to this work. For each constraint, we have presented several results on the capacity of the constrained channel, and efficient encoders/decoders, and provided the best options for the composite letters to obtain capacity-approaching codes with the lowest possible redundancy.

\end{document}